\documentclass[11pt]{article}
\usepackage[utf8]{inputenc}
\usepackage{amsfonts}
\usepackage{xcolor}         
\usepackage{amsthm}
\usepackage{amssymb}
\usepackage{amsmath}
\usepackage[ruled]{algorithm2e}
\usepackage{comment}
\usepackage{microtype}
\SetAlgoSkip{smallskip}
\usepackage{enumitem} 
\setlist{noitemsep,topsep=3pt,parsep=0pt,partopsep=0pt}
\usepackage{thmtools}
\usepackage{thm-restate}

\newif\ifsub
\subtrue

\usepackage[margin=1in]{geometry}

\newtheorem{theorem}{Theorem}[section]
\newtheorem{lemma}[theorem]{Lemma}
\newtheorem{claim}[theorem]{Claim}
\newtheorem{proposition}[theorem]{Proposition}
\newtheorem{corollary}[theorem]{Corollary}
\newtheorem{fact}{Fact}
\newtheorem{definition}{Definition}
\newtheorem{example}{Example}
\newtheorem{remark}{Remark}

\usepackage{hyperref}

\newcommand*{\email}[1]{%
    \normalsize\href{mailto:#1}{#1}\par
    }

\newcommand{\authnote}[2]{{ \footnotesize \bf{[#1: #2]~}}}
\newcommand{\vm}[1]{{\color{orange}\authnote{VM}{#1}}}
\newcommand{\pc}[1]{{\color{cyan}\authnote{PC}{#1}}}

\newcommand{\ignore}[1]{}

\ifsub
\renewcommand{\authnote}[2]{}
\fi

\renewcommand{\vec}[1]{\mathbf{#1}}
\newcommand{\vecp}{\vec{p}}

\newcommand{\lazy}{\bot}

\title{Blockchain Participation Games}
\author{Pyrros Chaidos\footnote{National and Kapodistrian University of Athens
and IOG, \email{p.chaidos@di.uoa.gr}} \and Aggelos Kiayias\footnote{University of Edinburgh and IOG, \email{aggelos.kiayias@ed.ac.uk}} \and Evangelos Markakis\footnote{Athens University of Economics and Business and IOG, \email{markakis@aueb.gr}}
}

\begin{document}
\maketitle

\begin{abstract}
    We study game-theoretic models for capturing participation in blockchain systems. Permissionless  blockchains can be naturally viewed as games, where a set of potentially interested users is faced with the dilemma of whether to engage with the protocol or not. 
    Engagement here implies that the user will be asked to complete certain tasks, whenever they are selected to contribute (typically according to some stochastic process) and be rewarded if they choose to do so.
    Apart from the basic dilemma of engaging or not, even more strategic considerations arise in settings  where users may be able to declare participation and then retract before completing their tasks (but are still  able to receive rewards) or are rewarded independently of whether they contribute. Such variations occur naturally in the blockchain setting due to the complexity of tracking ``on-chain'' the behavior of the participants.  
    
    We capture these participation considerations offering a series of models that 
    enable us to reason about 
    the basic dilemma, the case where  retraction effects influence the outcome and 
    the case when payments are given universally irrespective of the stochastic process. 
    In all cases we provide characterization results or necessary conditions on the structure of Nash equilibria. Our findings reveal that appropriate reward mechanisms can be used to stimulate participation and avoid negative effects of free riding, results that are in line but also can inform  real world blockchain system deployments.
\end{abstract}

\section{Introduction}

Blockchain protocols \cite{nakamoto} are typified by so called ``permissionless participation'', where the agents get to decide whether they wish to engage in the  protocol or not and if they choose to, they can do so {\em unilaterally}. This means that the system is capable of making the necessary adjustments to accommodate for increased or decreased participation, while there is no authority that whitelists the agents who participate --- for any user, merely downloading the software and running it is sufficient to become a part of the network. 

Based on the above, we observe that every running blockchain defines a {\em participation game}. A simple version of this game, can be described as follows: consider a protocol that operates in distinct units of time that we will call {\em epochs}. Imagine now a population of potentially interested agents, with a binary action space, who need to decide whether to engage (participate) or not. There are two prominent features that we are interested in studying in this work.
First, blockchain protocols such as Bitcoin \cite{nakamoto}, Algorand \cite{algorand}, Ouroboros \cite{ouroboros}, and Ethereum \cite{wood2014ethereum}, incorporate a stochastic process, where only some of the agents who chose to participate are {\em eligible} to contribute within an epoch, and the others are not. This is an essential component that is either achieved via so called, proof of work, or proof of stake or other similar techniques and ensures that the resulting  complexity of the protocol will be {\em sublinear} in the number of participating agents. Otherwise, it becomes very unlikely to have reasonable performance guarantees. 

The second feature is the {\em threshold} behavior of such systems, where in order for the blockchain protocol to make progress, the number of contributing users (among the eligible ones in each epoch) should exceed a certain {\em threshold} $k$. In some protocols, this threshold can be merely $1$ (e.g., in Bitcoin it is sufficient that one agent produces a block for the blockchain to advance), protocols, 
a larger $k$ is required (e.g., Ethereum currently needs a 2/3-majority voting among its randomly selected committees of at least 128 block validators in each epoch  \cite{UEthereumCommittees}). 




With respect to the utility of the agents, in most cases of interest, the protocol issues a reward to those who were both eligible and actually contributed (i.e., completing whatever task was dictated by the protocol for the eligible users) within an epoch, while at the same time, participating incurs a cost, incorporating effort, time and equipment. The utility then clearly depends on the stochastic process that determines the agents' eligibility.
In the simplest scenario, for example, all agents are treated equally and have the same probability of being eligible.
Finally, on top of rewards and costs, the 
protocol also induces a non-negative public benefit which is expressed as an additive ``bonus'' parameter, enjoyed by all the agents (including those who abstain) as long as the blockchain makes progress.
We stress that the game described so far is applicable not only for the process of producing new blocks from one epoch to another, but also for other applications within blockchain systems, where a group task needs to be completed, such as producing SNARKs (e.g., \cite{groth2016size}) for bootstrapping new users in the system or for building bridges between blockchains. 

The model described so far already gives rise to some interesting consequences, as elaborated in Section \ref{sec:simple}. If we delve into the implementations of such reward mechanisms however, there can be a significant burden imposed by keeping track of all participants who contributed in order to issue rewards. This may come in conflict with efficiency considerations and the use  of cryptographic primitives which  compress the participation information in order to offer complexities sublinear in the number of engaging parties. A concrete  example of this behavior in the context of blockchains is compact certificates~\cite{micali2021compact}), which carefully select what signatures to include when composing a multi-user certificate so that the certificate's size is kept small. 
In systems that utilize such more efficient primitives, it is impossible to reward exactly those who were eligible and participated and thus one has to resort to rewarding even those that may not be fully participating\footnote{An example of such a mechanism from the real world was staking rewards in Algorand up until April 2022, see e.g., \url{https://www.algorand.foundation/200-million-algo-staking-rewards-program}.}.

From the perspective of our work, such implementation considerations open up further strategic choices. 
In case  all eligible players are rewarded irrespective
of whether they complete all the assigned tasks, we have a participation
game {\em with retraction}, where  agents can declare they will participate, but afterwards refrain from completing all their assigned tasks. On the other hand, in case all players are rewarded 
irrespective of whether they were even eligible at an epoch, we have a participation game {\em with universal payments}. 
In these settings we observe a trade-off: the complexity of implementing the mechanism is lower (as the system does not need to keep track of detailed information regarding how players perform their assigned tasks) but 
 the possibility is opened up that some participants can become {\em free-riders}, reaping the rewards of an advancing blockchain without incurring costs to themselves. 
The  question we seek to answer here is whether the blockchain system remains viable and under what conditions. 
%

\subsection{Contribution}
Based on the previous discussion, we propose a formal framework to study participation games that focus on the aspects of engagement and free riding in blockchain systems. 
We start in Section \ref{sec:simple} with introducing a simple model, where each user is simply faced with the basic dilemma of participating or not and rewards are given accordingly. In Section \ref{sec:richer-model}, we then consider a richer model with the possibility of retraction, while in Section~\ref{sec:algorand-model} we extend our investigation to the setting of universal payments.  Our first observation is that in any attempt to define such games, the trivial profile where nobody participates is an equilibrium. We view this more as an artifact of the definition, and certainly far from what is observed in practice. We are therefore interested in the following questions.
\begin{quote}
\emph{Q1: Do these games possess other non-trivial Nash equilibria? If so, how big is the percentage of users that chooses to contribute at an equilibrium? }
\smallskip

\emph{Q2: How should we set the reward to the users so as to incentivize participation for a sufficiently large fraction of users so that the blockchain system remains in operation ?} 

\end{quote}

We start in Section \ref{sec:simple} with introducing a basic model for participation games, with the features we have outlined so far, and where each user faces the dilemma of participating or not with rewards given accordingly. We consider different variations of the game based on the selection probability, with an emphasis on the homogeneous case that treats all players equally. Our findings show that we can have equilibria with a very high level of engagement from the users, as long as the reward parameter is set within appropriate ranges, dependent on the cost and the other game parameters. We also extend our analysis to the non-homogeneous case, and even though the participation level may not always be as high as before, we can still show that it is high enough that the blockchain makes progress with high probability in every epoch.

Moving on, in Section \ref{sec:richer-model}, we consider a richer model with the possibility of retraction. This model requires a more involved analysis,  since users now have an additional choice of declaring participation and then ``retracting'' i.e., not completing all the assigned tasks (but still get rewarded when eligible). Even so, we are able to show that non-trivial equilibria will have a relatively high number of contributors. 

Finally, in Section~\ref{sec:algorand-model}, we consider the case of universal payments where participants are being paid irrespective of their eligibility. This respresents the simplest bookkeeping possible in terms of implementing the reward mechanism. In the context of universal payments, we consider both the basic participation game as well as the case of games with retraction. In both accounts, expectedly, the simplification in bookkeeping comes with the negatives of higher overall expenditure as well as the unfairness that stems from paying participants who do the actual work the same with those that skip their tasks. 

Overall, we believe our findings reveal a positive picture, confirming that simple reward mechanisms can stimulate engagement; this is consistent with what is observed in various actual blockchains, having reasonably large and active user populations, even with the possibility of retraction. 
Moreover our quantitative analysis for equilibria conditions can inform blockchain designers who wish to understand the level of rewards needed to make such systems viable and successful. 

As a roadmap to the paper we illustrate the different variants of our modeling  of participation games in table~\ref{tab:overview}.

\begin{table}[tbp]
\begin{center}
\begin{tabular}{|c|p{7cm}|c|c|c|}
\hline
Event & Intepretation  & Sec.~\ref{sec:simple} & Sec.~\ref{sec:richer-model} & Sec.~\ref{sec:algorand-model} \\ \hline
Abstain & Player chooses  not to participate & no reward & no reward & no reward \\ \hline
Participate & Player downloads and runs software & (depends) & (depends) & reward  \\ \hline
Eligible & The participating player is selected to contribute & reward & reward & reward \\ \hline
Not Eligible & The participating player is not selected & no reward & no reward & reward \\ \hline
Retract & The participating eligible player does not perform all assigned tasks & no reward & reward & reward \\ \hline
\end{tabular}
\caption{Overview of the different types of participation games showing different reward outcomes depending on the various possible events taking place.} \label{tab:overview}
\end{center}
\end{table}

\subsection{Related work}
The games that we introduce in Section \ref{sec:simple} bear some similarities with classic discrete public good games, also referred to as step-level games, which have been extensively studied, both theoretically (see e.g. \cite{PR84,GN90} and even more recent variants in \cite{GN22}), and experimentally, within behavioral game theory (indicatively see e.g., \cite{EP90}). The main difference is that we have a randomized process for determining the eligible players per epoch, whereas in public good games, any person is automatically eligible to contribute. We also pay special attention to the case where the public benefit is zero (for users who care only for the monetary reward). Moreover, our models in Section \ref{sec:richer-model}, and in Section \ref{sec:univeral+retraction} diverge further from public good games, as they have a richer strategy space.

There are already numerous works that have studied various 
game-theoretic aspects of blockchain systems. A common theme right from the outset of blockchains has been the study of {\em mining games}, cf.~\cite{mininggame},
where participants are facing the dilemma of which version of the ledger to extend. 
Important results in this context include models for selfish-mining, such as 
\cite{selfish}, as well as  \cite{kroll,mininggame,DBLP:conf/ec/FiatKKP19,10.1145/3490486.3538337}, 
that established both positive and negative
results on  whether the underlying  protocol is an equilibrium
or whether parties are incentivized to deviate from ``honest'' behavior and resort to block withholding.   
Another common theme has emerged from game theoretic aspects of 
{\em pooling} behavior: given that such protocols are
permissionless and they do not incorporate any mechanism capable 
of distinguishing the participants as being separate entities,
it is possible for them to form coalitions --- called mining pools or stake pools --- 
and act in tandem as a single entity in the protocol. 
Prior work established various conditions under which such pools arise
and studied their relative size~
\cite{DBLP:conf/fc/SchrijversBBR16,DBLP:conf/eurosp/BrunjesKKS20,DBLP:conf/aft/KwonLKSK19,DBLP:conf/innovations/ArnostiW19,kiasto2021,kiayias-sagt}. 

In blockchain systems users also compete for transaction processing time 
and space. Auction mechanisms can be used to increase the overall 
welfare  or the profit extracted by maintainers cf.~\cite{DBLP:journals/sigecom/Roughgarden21}. Moreover, given that ordering
and injecting transactions may increase the profits of maintainers, a ``miner extractable value''  (MEV) game may arise between maintainers and users~\cite{DBLP:conf/sp/DaianGKLZBBJ20}. 

A recent work by Motepalli and Jacobsen \cite{DBLP:journals/corr/abs-2104-05849} uses evolutionary game theory to study reward mechanisms in blockchains but under a very different model: in their theoretical model the actions of the user being analyzed do not affect the outcome of the game. Furthermore, they adopt a tighter focus on blockchain-related games with selfish actions modelling potentially malicious behavior such as validating invalid transactions. We opt for a more traditional model where selfish players simply try to avoid work. This has the benefit of making our analysis more general, as blockchain-specific measures against malicious behaviour may not translate well to other settings. \pc{was cut from camera}

We stress that the above studies on selfish mining and stake pools are conditioned 
on having sufficient participants
engaged  in the system. Hence, we view these approaches as orthogonal to ours, since our goal
is to study the basic dilemma of participating or not. The question we are interested in, is whether the participation game itself has favorable equilibria and under what conditions 
they occur.


\section{A basic model for participation games in permissionless protocols}
\label{sec:simple}


The basic model of participation games captures the  setting where a permissionless system invites agents to participate (a prospect that requires some expenditure on their part, e.g., run a software that processes transactions and attempts to reach consensus with other participants) 
and rewards them if they become eligible. 

Let $N=\{1, 2, \dots, n\}$ be a population of $n$ users. We consider protocols that run on a per epoch basis without loss of generality, and where the goal is to make ``progress'' in each epoch. Progress may be related to block production in longest chain protocols or it could even be related to other applications within blockchains, where some task needs to be carried out by the users (like issuing some group certificate, a checkpoint,
or validating a rollup). 
We focus on modeling applications with the following two features.

\begin{itemize}
    \item[(i)] The use of randomization. As it is impractical to ask all users to contribute, due to throughput and other practical considerations, the protocol selects in each epoch only a subset of users that are eligible to contribute and be rewarded. 
    
    \item[(ii)] The successful completion of the task is {\it threshold-based}, i.e., there is a public parameter $k$, so that progress is made  only if at least $k$ users contributed towards carrying out the task. 

\end{itemize}


We introduce now some of the relevant model parameters. In particular, we use the following abstractions:

\begin{itemize}
    \item Let $\alpha$ be the per-epoch cost of participation. This cost can be seen as the average per-epoch cost of running the protocol software uninterruptedly throughout an epoch, possibly updating it when necessary.
    This cost can capture both actual monetary cost (in electricity, etc.) and perceived human effort cost. 
    
    \item Let $r$ denote the monetary reward\footnote{We can think of $r$ as the  {\em expected} reward per epoch, conditioned on eligibility and on completion of the task. Using an expectation here enables us to capture the fact that even when players are eligible, they may occasionally lose the opportunity to contribute and be rewarded (e.g., in Bitcoin they can lose the race to distribute their block which is subsequently dropped by the network).} that is given to each player who is eligible (i.e., selected, according to the randomization procedure) in a given epoch, as long as progress is made. We assume here that the reward is given to a user after the protocol checks that indeed a user performed the actions necessary to contribute.   
    
    \item Let $k$ be the threshold, i.e., the minimum required number of eligible participants that need to contribute in an epoch, so that the blockchain makes progress.
    
    \item Let $q$ be the probability with which a player is eligible in a given epoch. This parameter may depend on $k$ and on the number of participating users. We assume an independent Bernoulli trial for each user, hence we do not insist that we have the same number of eligible players per epoch. Reasonable choices for $q$ is to make it sufficiently high so that in expectation we have enough eligible players per epoch to reach the threshold. Indicatively, we could think of $q$ in the range of $[2k/n, 3k/n]$. See also Remark \ref{rem:symmetric} below on the treatment of players with different stake or hashing power.
    \item Let $v$ be the inherent value that a player associates with the blockchain making progress. We note that this could be quite small compared to the blockchain rewards for maintenance but is intended to model the potential benefit the existence of a blockchain (or of the specific system that the players are invited to contribute to) brings to its users.  
\end{itemize}

\begin{remark}
\label{rem:symmetric}
    The model we are considering describes a homogeneous population, where both the probability of selection and the inherent value are the same for every player. We view this as an initial step on the analysis of such systems, and in Section \ref{sec:asym} we expand to cover asymmetric players in terms of the selection probability. 
\end{remark}

We assume that every user has two possible pure strategies in every epoch. The first is to simply abstain, and the second choice is for the user to participate. The latter means that the user chooses to be active throughout the epoch and also to proceed with completing the necessary tasks if  selected to do so. We analyze here players that when they choose to participate, they will not tamper with the software, and will abide by the protocol rules. We discuss in Section \ref{sec:richer-model} a richer model with ``retractors'' where users may also consider skipping their assigned tasks even though they expressed willingness to participate.

For brevity, let $\mathcal{S} = \{\lazy, \mathrm{P}\}$, be the binary strategy space of each player, where $\lazy$ means abstaining and $\mathrm{P}$ stands for choosing to participate. A strategy profile is a tuple $s = (s_1,\dots, s_n)\in \mathcal{S}^n$ specifying a choice for each player. As usual, given a profile $s$, and a player $i$, we denote by $s_{-i}$ the profile $s$, restricted to all players except $i$.

Our goal is to study the pure Nash equilibria of the static 1-epoch game. These are the profiles where the system can converge if the game is played repeatedly over many epochs. To define equilibria, we need first to define the (expected) utility that a player receives in a strategy profile. We outline first, in Table \ref{tab:tab-new-events}, the utility under the possible scenarios that may occur.

\begin{table}[tbp]
\begin{center}
\begin{tabular}{|c|c|c|}
\hline
Possible scenarios & Progress is made & No Progress \\ \hline
Abstain & $v$ & $0$   \\ \hline
Participate but not eligible& $v - \alpha$ & $-\alpha$ \\ \hline
Participate and eligible  & $r+v - \alpha$ & $- \alpha$ \\ \hline
\end{tabular}
\caption{Possible events and corresponding rewards in the basic model of a participation game. } \label{tab:tab-new-events}
\end{center}
\end{table}

\paragraph{Example.} As an illustration, we cast the above  game in the context of Bitcoin~\cite{nakamoto}. We consider one ``epoch'' to be the period during which a single block is added to the blockchain. The threshold can be taken to be $k=1$, since a single eligible participant is needed for the system to make progress. The probability $q$ corresponds to the relative hashing power of the miner. The cost $\alpha$ corresponds to the cost of running the mining equipment. Finally $r$ is determined by the block reward times the probability that an eligible miner will succeed in adding her block to the blockchain (which will be below $1$ due to the possibility of  other eligible miners disseminating competing blocks). 

\subsection{Equilibrium constraints}

To describe the constraints that need to hold in order to have an equilibrium, it is convenient to use the terminology introduced in the following definition.
\begin{definition}
Given a player $i$, consider a strategy profile $s_{-i}$, for all players except $i$. Then, let
\begin{itemize}
    \item $p(s_{-i})$ be the probability that progress is made in a given epoch, without taking into account what player $i$ does. This is equal to the probability that at least $k$ people out of the players who have selected to participate under the profile $s_{-i}$, are selected to be eligible. 
    \item $p(i, s_{-i})$ be the probability that progress is made in a given epoch, given that $i$ chose to participate, was eligible, and completed her task. This is equal to the probability that at least $k-1$ people out of the remaining participating players, except $i$, are eligible to contribute.
\end{itemize}
\end{definition}

Given a profile $s_{-i}$ for all players except $i$, the expected utility of player $i$, for each one of his pure strategies is described in Table \ref{tab:tab-new-utilities}.

\begin{table}[tbp]
\begin{center}
\begin{tabular}{|c|c|}
\hline
Action of pl. $i$ & Expected utility of pl. $i$, given $s_{-i}$   \\ \hline
$\lazy$ & $p(s_{-i})v$    \\ \hline
$P$ & $(1-q) p(s_{-i}) v + q p(i, s_{-i})(r+v) - \alpha$  \\ \hline
\end{tabular}
\caption{Expected utility under the possible events for a player $i$.} \label{tab:tab-new-utilities}
\end{center}
\end{table}

We can think of any strategy profile $s = (s_1,\dots, s_n)$, as partitioning the players into 2 sets, the set of possible contributors $C$, who are the people choosing to participate, and the set $A$ of abstainers.
Therefore, a profile $s$ is a Nash equilibrium if and only if

\noindent\begin{minipage}{0.5\linewidth}
\begin{eqnarray*}
    u_i(s) & \geq & u_i(\lazy, s_{-i}) ~~\forall i\in C
\end{eqnarray*}
\end{minipage}%
\begin{minipage}{0.5\linewidth}
\begin{eqnarray*}
    u_i(s)  & \geq & u_i(\mathrm{P}, s_{-i}) ~~\forall i\in A \label{eq:between-L-and-C} 
\end{eqnarray*}
\end{minipage}\par\vspace{\belowdisplayskip}

After substituting the expressions of Table \ref{tab:tab-new-utilities}, the above inequalities are equivalent to:
\begin{eqnarray}
    q\cdot [(r+v)p(i, s_{-i}) - vp(s_{-i})] & \geq & \alpha ~~\forall i\in C \label{eq:initial-from-C-to-L} \\
        q\cdot [(r+v)p(i, s_{-i}) - vp(s_{-i})] & \leq & \alpha ~~\forall i\in A \label{eq:initial-from-L-to-C} 
\end{eqnarray}

\begin{fact}
The profile where every player chooses to abstain is a pure strategy Nash equilibrium for $k>1$. We refer to it as the trivial equilibrium. 
\end{fact}

Given the above fact, the main question is whether there exist other non-trivial equilibria, and whether they achieve high levels of participation, which is the focus of the remaining section.

\subsection{Pure Nash equilibria for the case that $v=0$}
\label{sec:simple-v=0}
To further simplify  the game, suppose that $v=0$. Apart from serving as a simplification, we also find this case to be quite important from the perspective of a protocol designer. The reason is that one of the goals in the study of such systems is to identify how large should the monetary rewards be in order to incentivize users to engage with the protocol. When $v$ is large, users are already motivated to participate and we would need a lower reward to incentivize them, thus, the worst-case scenario in terms of upper bounding the total monetary rewards needed, is when $v=0$. 
Interestingly, this model falls into the broader class of games with {\emph strategic complements}, where the incentive for a player to take the "desirable" action has a monotonic behavior in terms of how many other people took the same action. 

\begin{proposition}
\label{prop:lattice}
    The family of games with $v=0$ exhibits strategic complements, as defined in \cite{Jackson08}, and hence its set of pure Nash equilibria forms a complete lattice.
\end{proposition}

We refer to \cite{Jackson08} for further discussion on strategic complements. Proposition \ref{prop:lattice} reveals a structural property on the set of equilibria, but does not give us any further information on their precise form. Hence, we continue by investigating the possible number of contributors that may arise. Suppose that there was an equilibrium with $|C| = \lambda$ contributors and $n-\lambda$ abstainers.
By expanding the probability terms in Equations \eqref{eq:initial-from-C-to-L} and \eqref{eq:initial-from-L-to-C}, we get the following  inequalities (since the same inequality has to hold for each contributor and ditto for each abstainer):


\begin{eqnarray}
q\cdot \sum_{j=k-1}^{\lambda-1}{\lambda-1 \choose j} q^{j} (1-q)^{\lambda-1-j} & \geq & \frac{\alpha}{r}  \label{eq:from-C-to-A} \\
q\cdot \sum_{j=k-1}^{\lambda}{\lambda \choose j} q^{j} (1-q)^{\lambda-j} & \leq & \frac{\alpha}{r} \label{eq:from-A-to-C} 
\end{eqnarray}

We elaborate further on how the above inequalities were derived. For \eqref{eq:from-C-to-A}, it has come from \eqref{eq:initial-from-C-to-L}, and with $v=0$, it is equivalent for a player $i$, to $rp(i, s_{-i}) \geq \alpha$. Now to calculate $p(i, s_{-i})$, one needs to consider all possible ways that the progress is made, given that $i$ has completed her task. This corresponds precisely to all the possible ways of selecting $k-1$ other players to be eligible, out of the $\lambda-1$ (excluding $i$) who have chosen to be in $C$. In an analogous manner, to calculate $p(i, s_{-i})$ in \eqref{eq:from-A-to-C}, we need to take into account all possible ways of selecting $k-1$ other players to be eligible, but now out of the $\lambda$ available contributors (since, for this case $i\in A$). 

The main result of this subsection is the characterization obtained in the following theorem, showing a sharp picture, that we can have at most two pure Nash equilibria. 

\begin{theorem}[Characterization]
\label{thm:simplemodel}
We can have at most two Nash equilibria as follows:
\begin{itemize}
    \item The trivial (all-out) profile $(\lazy, \lazy, \lazy)$, where nobody contributes, is a pure Nash equilibrium for $k>1$, or when $k=1$ and $q\leq {\alpha \over r}$.
    \item There is no equilibrium that has both a positive number of contributors and a positive number of abstainers.
    \item The all-in profile, where everybody participates, is an equilibrium if and only if:
\begin{equation}
 q\cdot \sum_{j=k-1}^{n-1}{n-1 \choose j} q^{j} (1-q)^{n-1-j} \geq \frac{\alpha}{r} \end{equation}
\end{itemize}
\end{theorem}

\begin{proof}[{\bf Proof of Theorem \ref{thm:simplemodel}}]
	The fact that the all-out profile is an equilibrium is trivial. For the all-in profile, we have that $C=N$ and hence, we only need to check that Equation \eqref{eq:from-C-to-A} holds when $\lambda=n$. But this is true precisely when the ratio $\alpha/r$ satisfies the stated bound.
	
	The most interesting part of the proof is to show that we cannot have any other pure equilibria. For the sake of contradiction, suppose that there is another equilibrium profile $s$, with $\lambda$ contributors and $n-\lambda$ users opting out, where $0< \lambda < n$. We should show that it is not possible to satisfy Equations \eqref{eq:from-C-to-A} and \eqref{eq:from-A-to-C} simultaneously. This is implied by the following lemma.

 \begin{lemma}
 \label{lem:no-intermediate}
     For every integer $\lambda$, with $0< \lambda < n$ it holds that
 $$\sum_{j=k-1}^{\lambda}{\lambda \choose j} q^{j} (1-q)^{\lambda-j}  = \sum_{j=k-1}^{\lambda-1}{\lambda-1 \choose j} q^{j} (1-q)^{\lambda-1-j}  + {\lambda-1 \choose k-2} q^{k-1} (1-q)^{\lambda-k+1}$$ 
 \end{lemma}
The proof of Lemma \ref{lem:no-intermediate} can be obtained as a special case of a more general result, namely Lemma \ref{lemma-asymnosameclass}, that we use in Section \ref{sec:asym-proof}. We therefore refer the reader to that proof. 
\end{proof}

We refer to the all-out profile as the trivial equilibrium.
Despite the existence of such an undesirable equilibrium, we do not view this as a disastrous property. In particular, the all-out profile is hard to be sustained over time, as it is easy to see that any coalition of at least $k$ users could deviate and gain more, and this also agrees with what we observe in practice. On the contrary, the all-in profile has much more attractive properties against coalitional deviations.
\begin{fact}
\label{fact:strong-eq}
    Whenever the all-in profile is an equilibrium, it is also a strong equilibrium, hence no coalition has a profitable deviation.
\end{fact}

Reflecting upon Theorem \ref{thm:simplemodel} and Fact \ref{fact:strong-eq}, we view these as positive news since they show that as long as the reward is sufficiently high, it is possible to incentivize all users to participate.

\begin{remark}
    By being able to enforce, via the reward, that all rational users participate, we can also tolerate a relatively high number of malicious or indifferent users, who may never wish to participate. In particular, even if we have close to $n/2$ such users, the remaining ones can still be motivated to participate, and therefore be able to make progress with high probability, when e.g., $q\geq 2k/n$. 
\end{remark}

\subsection{Analysis when $v>0$}
Coming now back to having a positive inherent value, suppose that $v>0$.
It is natural to expect that the inherent value is typically smaller than the monetary reward $r$. 
The presence of the value $v$ introduces some differences with the previous subsection and we can no longer have such a sharp characterization as in Theorem \ref{thm:simplemodel}. Nevertheless, for non-trivial equilibria, we will still have a relatively high number of contributors, as described in the necessary condition below.

\begin{theorem}
\label{thm:v>0}
    When $v\leq r$, then any non-trivial equilibrium must have at least $(2-q)\frac{k-1}{q}$ contributors.
\end{theorem}
\begin{proof}
Suppose that there exists an equilibrium with $|C|=\lambda$ and $|A| = n-\lambda$, and with $0< \lambda < n$. If there was an equilibrium with both contributors and free riders, by expanding the equilibrium conditions \eqref{eq:initial-from-C-to-L} and \eqref{eq:initial-from-L-to-C}, and rearranging terms, the following would have to hold.

\begin{eqnarray}
r\cdot \sum_{j=k-1}^{\lambda-1}{\lambda-1 \choose j} q^{j} (1-q)^{\lambda-1-j} + v\cdot {\lambda-1 \choose k-1} q^{k-1} (1-q)^{\lambda-k} & \geq & \frac{\alpha}{q}  \label{eq:with-v-from-C-to-A} \\
r\cdot \sum_{j=k-1}^{\lambda}{\lambda \choose j} q^{j} (1-q)^{\lambda-j} + v\cdot {\lambda \choose k-1} q^{k-1} (1-q)^{\lambda-k+1} & \leq & \frac{\alpha}{q} \label{eq:with-v-from-A-to-C} 
\end{eqnarray}

For brevity, let us rewrite \eqref{eq:with-v-from-C-to-A} as $r\cdot \Sigma_1 + v\cdot f_1\geq \alpha/q$, and similarly, we can rewrite \eqref{eq:with-v-from-A-to-C} as $r\cdot \Sigma_2 + v\cdot f_2\leq \alpha/q$. Given the form of these inequalities, if we would show that 
$r\cdot \Sigma_2 + v\cdot f_2 > r\cdot \Sigma_1 + v\cdot f_1$, we would get a contradiction, and hence there cannot exist any equilibrium with $\lambda$ contributors, for the $\lambda$ we started with. 
This condition is equivalent to:
$$ r(\Sigma_2 - \Sigma_1) > v(f_1-f_2)$$

To proceed, we note that it is a simple calculation to verify that when $\lambda < \frac{k-1}{q}$, we have that $f_1<f_2$, and we get already the desired contradiction. Therefore, at a non-trivial equilibrium, it must hold that $\lambda \geq \frac{k-1}{q}$. In fact, we can  obtain an even stronger lower bound for $\lambda$. Given that $v\leq r$, to obtain a contradiction it suffices that 
\begin{equation} 
\label{eq:contr}
\Sigma_2 - \Sigma_1 > f_1-f_2
\end{equation}

Consulting Lemma \ref{lem:no-intermediate}, we can see that  $\Sigma_2 - \Sigma_1$,   is actually 
$$ \Sigma_2 - \Sigma_1 = {\lambda-1 \choose k-2} q^{k-1} (1-q)^{\lambda-k+1}$$
By plugging this in \eqref{eq:contr}, and by substituting $f_1, f_2$ and simplifying the resulting expression, we obtain the claimed bound for $\lambda$. 
\end{proof}

Theorem \ref{thm:v>0} provides a necessary condition for $\lambda$, but it does not tell us for which values of $\lambda$ we do have an equilibrium.
This is dependent on the other parameters, such as $\alpha$ and $v$. However, it is possible to check for every $\lambda$ if there exists a range for the reward $r$ as a function of $\alpha, v$ and $k$, so that we can have an equilibrium with $\lambda$ contributors. Additionally , we point out that we can always set the reward appropriately so as to incentivize all players to participate.

\begin{theorem}
    There always exists a non-empty range for the reward $r$ so that the all-in profile is an equilibrium. Namely this holds as long as 
    $$r \geq \frac{\frac{\alpha}{q} - v{\lambda-1 \choose k-1} q^{k-1} (1-q)^{\lambda-k}}{\sum_{j=k-1}^{\lambda-1}{\lambda-1 \choose j} q^{j} (1-q)^{\lambda-1-j}} $$
    
\end{theorem}

Note that the smaller the nominator above, the looser the constraint for $r$, and as $v$ drops down to 0, this is when we get stricter constraints for $r$. To summarize, when $v>0$, it is conceivable that we have equilibria other than the all-in and the all-out profile, in contrast to Section \ref{sec:simple-v=0}. However, all the non-trivial equilibria have a sufficiently high number of participants. For reasonable choices of $q$ in the range of $[2k/n, 3k/n]$, as alluded to in the beginning of Section \ref{sec:simple}, Theorem \ref{thm:v>0} implies that we can have equilibria with more than $n/2$ participants, when the reward is set appropriately. 

\subsection{Equilibria for the non-symmetric case w.r.t. selection probability}\label{sec:asym}
We move now to study the asymmetric case in terms of the selection probability. In particular, suppose that each player now has a possibly different probability $q_i$ of being selected. We will stick to the case where $v=0$ for simplicity.

The equilibrium constraints are now more complex because different  contributors have a different probability of being  selected. In  an equilibrium with a set $C$ of $\lambda$ contributors and a set $A$ of $n-\lambda$ abstainers, we must have:

\begin{eqnarray}
q_i\cdot \sum_{m=k-1}^{\lambda-1} ~\sum_{S\subseteq C\setminus \{i\}, |S|=m} ~\left( \prod_{j\in S} q_j \cdot \prod_{j\in C\setminus S\cup \{i\} } (1-q_j) \right) & \geq & \frac{\alpha}{r} ~~~\forall i\in C
\label{eq:asym-C-to-A}\\
q_i\cdot \sum_{m=k-1}^{\lambda} ~\sum_{S\subseteq C, |S|=m} ~\left( \prod_{j\in S} q_j \cdot \prod_{j\in C\setminus S} (1-q_j) \right) & \leq & \frac{\alpha}{r} ~~~\forall i\in A \label{eq:asym-A-to-C}
\end{eqnarray}

Despite the added complexity, it is still feasible to understand the composition of contributors and abstainers at equilibrium. 
Proposition \ref{prop:lattice} holds for the asymmetric setting as well, hence the equilibrium set forms a complete lattice. Furthermore, 
the main theorem of this subsection gives a characterization of the possible equilibrium structures that can arise.

\begin{theorem}
\label{thm:structure-char}
Given a game with $n$ players, let $q_1 \geq q_2 \geq \dots \geq q_\ell$ be the distinct selection probabilities.
Then for any non-trivial equilibrium, there must exist a threshold $q_e$ so that all players with $q_i \geq q_e$ are contributors, and they are at least $k$, and all players with $q_i < q_e$ are abstainers.
\end{theorem}
\begin{proof}
    First, we show that players with the same selection probability must select the same action at equilibrium (Lemma \ref{lemma-asymnosameclass}). 
Second, due to Lemma \ref{lemma-asymnoswaps}, in an equilibrium, there can be no contributor with smaller selection probability than \textit{any} non-contributor. Thus, in any equilibrium with both contributors and non-contributors, we must have that the ``poorest'' contributor is strictly richer than the ``richest'' non-contributor. It must also be the case that if  we do have contributors they must number at least $k$, as otherwise their reward would be negative.\end{proof}

Theorem \ref{thm:structure-char} shows that it is conceivable to have equilibria with a relatively low  number of contributors in the asymmetric case. Nevertheless, any non-trivial equilibrium will have at least $k$ contributors. Moreover, since in an actual proof of stake blockchain system, the players with higher selection probabilities are expected to be the ones who possess higher amounts of cryptocurrency, it shows us that the richer players will ``do their duty'' and choose to contribute.

Finally, note that the above theorem is a necessary condition tells us how the equilibria can look like, but it does not identify actual equilibrium profiles nor does it give any information on whether we can have multiple equilibria with different thresholds. 
This will necessarily depend on the range of the ratio $\alpha/r$. Fortunately, we can efficiently investigate all possible equilibria via a small number of checks.

\vm{Pyrro, check this}\pc{checked, added proof}
\begin{theorem} \label{thm:structure-equiv}
Let $q_1 \geq q_2 \geq \dots \geq q_\ell$ be the distinct selection probabilities. For every 
$a\in[\ell]$, there exists an equilibrium with the threshold for the contributors being $q_a$ if and only if \eqref{eq:asym-C-to-A} is satisfied by using $q_i=q_a$ and \eqref{eq:asym-A-to-C} is satisfied when we use $q_i=q_{a+1}$. As for the all-in profile, there exists a non-empty range for $\alpha/r$ that makes it an equilibrium.
\end{theorem}
\begin{proof} We begin with the second clause. Set all player to be contributors and set $r$ so that equation \eqref{eq:asym-C-to-A} holds for players with $q_i=q_\ell$. Then, by Lemma \ref{lemma-asymcontrib} the same equation will also hold for all other players and we will have an equilibrium. For the first clause, we begin with the first direction: suppose \eqref{eq:asym-C-to-A} is satisfied by using $q_i=q_a$, and \eqref{eq:asym-A-to-C} is satisfied when we use $q_i=q_{a+1}$. Then, by Lemma \ref{lemma-asymcontrib} eq. \eqref{eq:asym-C-to-A} is also satisfied for all players with $q_i\geq q_a$, and by Lemma \ref{lemma-asymnoncon} eq. \eqref{eq:asym-A-to-C} is satisfied for all players with $q_i \leq q_{a+1}<q_a$. This accounts for all users, and describes an equilibrium in which only users with $q_i\geq q_a$ contribute. 
For the converse, suppose there exists a non-trivial equilibrium. By Theorem \ref{thm:structure-char}, it is characterized by a threshold $q_e$. If $q_e=q_j$ for $j\in[\ell]$, the statement is true. 
Else, if $q_e>q_1$, the equilibrium is trivial. Thus there must exist some  $j\in[\ell]$ so that $q_e \leq q_{j}$. Let $j^*$ be the maximum such index. We set $q_e'=q_{j^*}$. We observe that the new threshold describes the same equilibrium: there exist no players with probability in the interval $[q_e,q_{j^*})$ so the truth value of  ($q_i \geq q_e$) is identical to ($q_i \geq q_{j^*})$ for all $i \in [\ell]$.

\end{proof}

The following examples provide some further intuition about these results.

\begin{example}
    Suppose we have a 4 player game with $q_1=q_2={1\over2}$ and $q_3=q_4={1\over4}$. Let $\alpha=1,r=5$ and $k=2$. Then there exist 3 different equilibria. $C= \emptyset$ i.e. ``all out'' is trivially an equilibrium as no player can reach $k=2$ alone. $C=\{1,2\}$ is also an equilibrium. The LHS of  \eqref{eq:asym-C-to-A} is ${1\over 4} > {1 \over 5}$ for the contributors. 
    For the non-contributors, the LHS of \eqref{eq:asym-A-to-C} is ${1\over 4} \cdot {3\over 4} = {3\over 16}$ which is less than ${1\over 5}$: if one of the poor players abstains, the frequency of reward is not enough for the other poor player to participate. 
    Finally, the ``all in'' profile $C=\{1,2,3,4\}$ is also an equilibrium. Checking the   LHS of  \eqref{eq:asym-C-to-A}, we have ${1\over 4} \cdot [1-{1\over2}\cdot{1\over2}\cdot{3\over4}]={13\over 64} > {1\over 5}$. 
\end{example}

We also remark that is is not always the case that we have an equilibrium for every threshold. 


\begin{example}
    Suppose we modify the above game to  have 4 players with $q_1=q_2={1\over3}$, $q_3=q_4={1\over4}$ and $r=8$. Then only $C= \emptyset$ and  $C=\{1,2,3,4\}$ are equilibria. The ``rich only'' profile $C=\{1,2\}$ is no longer an equilibrium. Checking the   LHS of  \eqref{eq:asym-C-to-A}, we have ${1\over 3} \cdot [1-{2\over3}]={1\over 9} < {1\over 8}$. The ``all in'' profile is still an equilibrium as (for the poorest contributor, $q_4$) the LHS of \eqref{eq:asym-C-to-A} is ${1\over4}\cdot\left( 1- {{2\over3}\cdot{2\over 3}\cdot {3\over 4} } \right)={1\over6} > {1\over 8}$.
\end{example}

\subsubsection{Supporting Lemmas for Theorems \ref{thm:structure-char} and \ref{thm:structure-equiv}}
\label{sec:asym-proof}
The proof of the theorems is obtained via a series of lemmas.



\begin{restatable}{rlm}{asymnosameclass}
\label{lemma-asymnosameclass}
Consider two players a and b, with $q_a = q_b$. At an equilibrium, these two players cannot choose different actions.  
\end{restatable}

Due to space, we defer the proof to  Appendix \ref{appendixproofasymnosameclass}. Combined with the next lemma they complete the proof of Theorem \ref{thm:structure-char}.

\begin{lemma} \label{lemma-asymnoswaps}   
Consider two players a and b, and let $q_b > q_a$. 
At an equilibrium, if player $a$ has chosen to be a contributor, player $b$ must also be a contributor, i.e. $b$ would have an incentive to deviate if she chooses to abstain. \end{lemma}

The proof is identical to that of Lemma \ref{lemma-asymnosameclass}, with the difference of setting  $q_b>q_a$. 

The next two lemmas are needed for the proof of Theorem \ref{thm:structure-equiv}.

\begin{lemma} \label{lemma-asymcontrib}
Consider a strategy profile where two players a and b have chosen to be in $C$, and let $q_a > q_b$. Then if player $b$ has no incentive to deviate (i.e., it satisfies \eqref{eq:asym-C-to-A}), the same holds for player $a$ as well.
\end{lemma}

\begin{proof}
Fix the 2 players $a, b\in C$ with $q_a> q_b$.
Suppose that player $b$ has no incentive to deviate, i.e., \eqref{eq:asym-C-to-A} is satisfied for $q_i = q_b$. Let $L_i=\sum_{m=k-1}^{\lambda-1} ~\sum_{S\subseteq C\setminus \{i\}, |S|=m} ~\left( \prod_{j\in S} q_j \cdot \prod_{j\in C\setminus S\cup \{i\} } (1-q_j) \right)$.
Then the difference of the LHS of \eqref{eq:asym-C-to-A} for $i = a$ minus the LHS of \eqref{eq:asym-C-to-A} for $i = b$, is $\Delta_{a,b}=q_a \cdot L_a - q_b \cdot L_b$. It suffices to show that 
$\Delta_{a,b} \geq 0$.


We first consider the term $q_a \cdot L_a$. We will rewrite it in terms of $q_b$ as follows:
\begin{eqnarray*}
&q_a \cdot  \sum_{S\subseteq C\setminus\{a\}, |S|\geq {k-1}}  &~\left( \prod_{j\in S} q_j \cdot \prod_{j\in (C\setminus \{a\} \setminus S ) } (1-q_j) \right) = \\ 
&q_a \cdot  \sum_{S\subseteq C\setminus\{a\}, |S|\geq {k-1},b\in S}  &~\left( \prod_{j\in S} q_j \cdot \prod_{j\in (C\setminus \{a\} \setminus S ) } (1-q_j) \right) \allowbreak \\
+ &q_a \cdot  \sum_{S\subseteq C\setminus\{a\}, |S|\geq {k-1}, b\notin S} &~\left( \prod_{j\in S} q_j \cdot \prod_{j\in (C\setminus \{a\} \setminus S ) } (1-q_j) \right) = \\
 &q_a \cdot q_b \cdot  \sum_{S\subseteq C\setminus\{a,b\}, |S|\geq {k-2}} & ~\left( \prod_{j\in S} q_j \cdot \prod_{j\in (C\setminus \{a,b\} \setminus S ) } (1-q_j) \right) \allowbreak \\
+ &  q_a \cdot  (1-q_b) \cdot  \sum_{S\subseteq C\setminus\{a,b\},|S|\geq {k-1}, } & ~\left( \prod_{j\in S} q_j \cdot \prod_{j\in (C\setminus \{a,b\} \setminus S ) } (1-q_j) \right) = \\
&q_a \cdot \left( q_b \cdot \Sigma_1 + (1-q_b) \cdot \Sigma_2 \right)&
\end{eqnarray*}

In the last line above, we have used for brevity $\Sigma_1$ and $\Sigma_2$ for the terms that are multiplied with $q_aq_b$ and $q_a(1-q_b)$ respectively. We also note that both terms $\Sigma_1$ and $\Sigma_2$ are symmetrical wrt $a,b$.  
If we now repeat the above calculations for $q_b \cdot L_B$ we obtain that it is equal to $q_b \cdot \left(q_a \cdot \Sigma_1 + (1-q_a) \cdot \Sigma_2\right)$.
Thus, $\Delta_{a,b}= q_a\cdot \left(q_b \cdot \Sigma_1 + (1-q_b) \cdot \Sigma_2\right)  - q_b \cdot \left(q_a \cdot \Sigma_1 + (1-q_a) \cdot \Sigma_2\right)$.

Simplifying, we have that $\Delta_{a,b}=\Sigma_2 (q_a-q_a \cdot q_b -q_b +q_a\cdot q_b ) = \Sigma_2 (q_a-q_b)$, which is positive.
\end{proof}

The next lemma shows an analogous statement for abstainers
\begin{lemma}\label{lemma-asymnoncon}
Consider a strategy profile where two players a and b have chosen to be in $A$, and let $q_a > q_b$. Then if player $a$ has no incentive to deviate (i.e., it satisfies \eqref{eq:asym-A-to-C}), the same holds for player $b$ as well.  
\end{lemma} 
\begin{proof}
We rewrite \eqref{eq:asym-A-to-C} as $q_i \cdot M_i \leq {\alpha\over r}$, where $M_i= \sum_{m=k-1}^{\lambda} ~\sum_{S\subseteq C, |S|=m} ~\left( \prod_{j\in S} q_j \cdot \prod_{j\in C\setminus S} (1-q_j) \right)$. 

By direct calculation, for two abstaining players $a$,$b$, it holds that  $M_b= M_a$. Thus, if $q_a > q_b$, then $q_a \cdot M_a > q_b \cdot M_b$, and by extension if $M_a \leq  {\alpha\over r}$ then $M_b \leq  {\alpha\over r}$.
\end{proof}

\section{A richer model: Participation games with retraction}
\label{sec:richer-model}


In this section we study a richer model motivated by two considerations. First of all, 
in blockchain systems, it can be inefficient to keep a record of which users among the eligible ones in a given epoch participated, so as to reward only them. For instance, 
using methods such as 
compact certificates \cite{micali2021compact}, 
threshold signatures \cite{desmedt1987society},
stake-based threshold multisignatures \cite{mithril}
or 
succinct non-interactive arguments of knowledge (SNARKs) like \cite{groth2016size}
 it is efficient to verify that sufficiently many  users participated, \emph{without} producing a list of the users who actually did: the added size of such a list runs contrary to the goal of efficiency of the underlying primitive. 
 Hence one mechanism to consider in such a setting is to provide the reward $r$ to everybody who was eligible in a given epoch without checking if they actually contributed. 

The second consideration is the fact that the
 assigned tasks whenever a user is eligible 
can be performed partially, with users being able to participate, become eligible but (via using modified software) avoid completing all the tasks that are requested by the protocol hoping that others will perform them.  
(the model of Section \ref{sec:simple} can be interpreted as setting this additional cost to be zero, i.e., if one chooses to participate this also means that they have to complete the tasks). 
%
%

The combination of the two above  features,  creates  further strategic considerations. The users are now able to extend their strategy space by the following option: choose to participate in the system (e.g., to check eligibility, so as to receive a reward whenever eligible) but not complete all the tasks (if they feel enough of the other users will do it). This is an undesirable scenario of free riding that may arise in practice, leading to slow or unreliable operation of the entire system. 
One simple approach to counteract this is to penalise users who fail to do their job, 
(for instance,  in Ethereum 2.0 there are penalties for lack of participation). 
%
%
We observe that if we indeed strip the reward $r$ from users who try to free ride, then their
strategic choices fall back to the model of Section \ref{sec:simple} and hence our analysis applies directly.
%
Moreover, the downside of a penalty mechanism would be the need to perform
bookkeeping of all those who engage to completion - something that counteracts
the efficiency benefit we were hoping to obtain by simplified bookkeeping as described above. 

The above considerations motivates us to investigate whether 
we can still get an effective mechanism that does not 
rely on tracking user behaviour. 
Specifically, the case where we extend
the reward $r$ to all eligible users, so 
that we only need to know whether progress was made or not. In this setting, users can be paid simply by proving their eligibility without the system keeping track of whether they completed all assigned tasks. 


We highlight below the similarities and differences between the new model and that of Section \ref{sec:simple}:

\begin{itemize}
    \item As before, $\alpha$ is the average cost for a user of running and maintaining the protocol software throughout an epoch, including  the per-epoch check of whether the user is eligible. Furthermore, the parameters $k, q, v$,  have the same interpretation as in Section \ref{sec:simple}.
    \item Let $\beta$ be  the additional cost incurred by an eligible player of bringing to completion the assigned tasks, where the contribution per player is not monitored by the accounting mechanism of the system (i.e., they system will not be  able to detect whether a player contributed to the completion of these tasks). 
    For comparison, in Section~\ref{sec:simple} it holds that $\beta=0$.
    \item Under the new model, the monetary reward $r$ can be claimed by all eligible players of an epoch, as long as progress is made, {\bf regardless} of whether they contributed the additional tasks or not.
\end{itemize}

Every player has now 3 possible pure strategies: 
\begin{itemize}
    \item Abstain.
    \item Declare to participate, and if eligible, contribute all assigned tasks.
    \item Retraction: declare to participate but if eligible, do not contribute any tasks (that can be avoided). 
\end{itemize}

The possible scenarios that can occur, and the corresponding utility, depending on the other players' decisions as well, are shown in Tables \ref{tab:tab-events} and \ref{tab:tab-utilities}. 

\begin{table}[tbp]
\begin{center}
\begin{tabular}{|c|c|c|}
\hline
Scenarios & Progress is made & No Progress \\ \hline
Abstain & $v$ & $0$   \\ \hline
Declares participation, not eligible & $v - \alpha$ & $-\alpha$ \\ \hline
Declares participation, eligible and completes tasks  & $r+v - \alpha-\beta$ & $- \alpha-\beta$ \\ \hline
Declares participation, eligible and retracts & $r+v - \alpha$ & $- \alpha$ \\ \hline
\end{tabular}
\caption{Possible events and corresponding rewards in a participation game with retraction.} \label{tab:tab-events}
\end{center}
\end{table}
 

\subsection{Equilibrium constraints}

As was the case in Section \ref{sec:simple}, the all-out profile is again trivially an equilibrium. We proceed to examine what other equilibria may exist in this game. As we will see, we do get a more interesting structure for the equilibria, compared to Section \ref{sec:simple-v=0}. 

We start with identifying the conditions that need to hold in order for a profile $s$ to be an equilibrium. 
Given a profile $s_{-i}$ for all players except $i$, the expected utility of player $i$, for each one of her pure strategies is described below.

\begin{table}[tbp]
\begin{center}
\begin{tabular}{|l|c|}
\hline
Action of player $i$ & Expected utility of pl. $i$, given $s_{-i}$   \\ \hline
Abstain & $p(s_{-i})v$    \\ \hline
Participate, contribute if eligible & $(1-q)p(s_{-i})v + q[p(i, s_{-i})(r+v) -\beta] - \alpha$  \\ \hline
Participate, don't contribute & $(1-q)p(s_{-i})v + qp(s_{-i})(r+v) - \alpha$\\ \hline
\end{tabular}
\caption{Expected utility under the possible events for a player $i$.} \label{tab:tab-utilities}
\end{center}
\end{table}

Any strategy profile $s = (s_1,\dots, s_n)$ partitions the players into 3 sets, the set of possible contributors $C$, who are the people choosing to participate and complete their task whenever selected to do so, the set of free-riders $F$, who choose to participate (so that they can get a reward, whenever eligible), but will not contribute, and the set $A$ of abstainers.

We group the equilibrium constraints into three groups, based on the possible deviations of the players. First of all, for a player $i\in C$, who decided to contribute if eligible, she should not have an incentive to detract and move to $F$. Symmetrically, a player from $F$ should not have an incentive to contribute, if eligible.
After simplifying these inequalities, based on Table \ref{tab:tab-utilities}, these are equivalent to:
\begin{eqnarray}
p(i, s_{-i}) - p(s_{-i}) & \geq & \frac{\beta}{r+v} ~~\forall i\in C \label{eq:from-C-to-G}  \\
p(i, s_{-i}) - p(s_{-i}) & \leq & \frac{\beta}{r+v} ~~\forall i\in F \label{eq:from-G-to-C} 
\end{eqnarray}
Intuitively, this means that a player $i\in C$ should be ``critical enough", i.e., there should be a lower bound on the difference between the success probabilities, i.e., between the probability that progress is made when $i$ is eligible and contributes, and the probability that the progress is made via the remaining players. This lower bound is independent of $\alpha$ but has to depend on $\beta$.

In a similar fashion, the players from $C$ should also not have an incentive to abstain, and move to $A$, and at the same time, the players from $A$ should not have an incentive to participate and contribute. This yields two more inequalities, which after simplifications are equivalent to:
\begin{eqnarray}
    q\cdot[(r+v)p(i, s_{-i}) - vp(s_{-i})] & \geq & \alpha+ \beta \cdot q ~~\forall i\in C \label{eq:from-C-to-L} \\
        q\cdot [(r+v)p(i, s_{-i}) - vp(s_{-i})] & \leq & \alpha+ \beta \cdot q ~~\forall i\in A \label{eq:from-L-to-C} 
\end{eqnarray}

Finally, for a player $i\in F$, she should not have an incentive to abstain. Its counterpart is that players from $A$ should also have no incentive to become free riders. This yields the following:

\begin{eqnarray}
q\cdot r\cdot p(s_{-i}) & \geq & \alpha  ~~\forall i\in F \label{eq:from-G-to-L} \\
q\cdot r\cdot p(s_{-i}) & \leq & \alpha  ~~\forall i\in A \label{eq:from-L-to-G} 
\end{eqnarray}


Summarizing, a strategy profile with a non-empty set of contributors, free riders and abstainers is a Nash equilibrium if and only if it satisfies the inequalities \eqref{eq:from-C-to-G} to \eqref{eq:from-L-to-G}. For equilibrium profiles where at least one of the pure strategies is not chosen by any player, one needs to restrict to the corresponding subset of inequalities among \eqref{eq:from-C-to-G} to \eqref{eq:from-L-to-G}.

\subsection{Equilibrium analysis when $v=0$}

As in Section \ref{sec:simple}, we focus on the case of zero intrinsic value for progress being made, i.e., $v=0$. This case is already technically much more involved than its corresponding counterpart in Section \ref{sec:simple-v=0}. Furthermore, it also serves as a sufficient illustration of the differences in the type of equilibria that may arise when free riding is present. We comment further on this at the end of the subsection.

We can already draw some initial conclusions by Equations \eqref{eq:from-C-to-G} to \eqref{eq:from-L-to-G}. Suppose first that we want to check if there exists a non-trivial equilibrium with a non-empty set of competitors and also with a non-empty set of free-riders and a non-empty set of abstainers. We show that this is false.

\begin{theorem}
\label{thm:no-C-and-A}
    There cannot exist an equilibrium where both $C\neq\emptyset$ and $A\neq \emptyset$. 
\end{theorem}
\begin{proof}
Suppose that there was such an equilibrium, with $C\neq\emptyset$ and $A\neq \emptyset$, and say $|C|=\lambda$. 
This requires that both Equations \eqref{eq:from-C-to-L} and \eqref{eq:from-L-to-C} need to hold.
With $v=0$, if we expand these equations by calculating the terms $p(i, s_{-i})$ and $p(s_{-i})$, we will obtain two inequalities that are very similar to \eqref{eq:from-C-to-A} and \eqref{eq:from-A-to-C} (with the only difference being that $\alpha$ is replaced by $\alpha + \beta q)$). But then, Lemma \ref{lem:no-intermediate} can be used again and obtain that there does not exist any equilibrium, with $\lambda$ contributors and a non-empty set of abstainers, with $0 < \lambda < n$. 
\end{proof}

By Theorem \ref{thm:no-C-and-A}, we have that for an equilibrium with a positive number of contributors, either it is the all-in profile or it can also contain some free riders but no abstainers. 
Note also that in case $C=\emptyset$, we can only have the trivial all-out equilibrium (when $k>1$).
Hence, in the sequel, we will examine the existence of free riders in equilibria.

\subsubsection{Equilibria with only contributors and free riders.}

We rewrite the equilibrium constraints, that need to hold for an equilibrium with $\lambda$ contributors and $n-\lambda$ free riders, where $\lambda>0$ and $\lambda < n$. 
From the equilibrium constraints \eqref{eq:from-C-to-G} - \eqref{eq:from-L-to-G}, we need only \eqref{eq:from-C-to-G}, \eqref{eq:from-G-to-C}, \eqref{eq:from-C-to-L} and \eqref{eq:from-G-to-L}, since we have no abstainers, and all we need is to ensure that contributors have no incentive to move to $F$ or $A$ and free riders have no incentive to move to $C$ or $A$. 
We need first to calculate the relevant success probabilities $p(i, s_{-i})$ and $p(s_{-i})$ in these inequalities, and we can do it in a similar manner as in Section \ref{sec:simple-v=0} (in the derivation of \eqref{eq:from-C-to-A} and \eqref{eq:from-A-to-C}). Namely, for $p(i, s_{-i})$ and $p(s_{-i})$, we will need to consider all possible ways that progress can be made in each case.
After carrying out these calculations, the existence of equilibria is equivalent to the following system of inequalities, which are equivalent to \eqref{eq:from-C-to-G}, \eqref{eq:from-G-to-C}, \eqref{eq:from-C-to-L} and \eqref{eq:from-G-to-L} respectively. 

\begin{eqnarray}
{\lambda-1 \choose k-1} q^{k-1} (1-q)^{\lambda-k}& \geq & \frac{\beta}{r} ~~ \label{eq:new'from-C-to-G}  \\
{\lambda \choose k-1} q^{k-1} (1-q)^{\lambda-k+1}& \leq & \frac{\beta}{r} ~~ \label{eq:new'from-G-to-C} \\
r\cdot q\cdot \sum_{j=k-1}^{\lambda-1}{\lambda-1 \choose j} q^{j} (1-q)^{\lambda-1-j}& \geq & \alpha + \beta q ~~ \label{eq:new'from-C-to-L}  \\
r\cdot q\cdot \sum_{j=k}^{\lambda}{\lambda \choose j} q^{j} (1-q)^{\lambda-j}& \geq & \alpha ~~ \label{eq:new'-L-and-G} 
\end{eqnarray}

The above system already yields some positive news regarding participation, as we can have the following lower bound on the number of contributors.

\begin{claim}
\label{cl:lambda-bound}
At a non-trivial equilibrium, with $\lambda$ contributors and $n-\lambda$ free riders, it must hold that
$$n-1 \geq \lambda \geq \frac{k-1}{q} $$
\end{claim}
\begin{proof}
Obviously $\lambda$ cannot exceed $n-1$. The lower bound on $\lambda$ is a direct consequence of using \eqref{eq:new'from-C-to-G} and \eqref{eq:new'from-G-to-C} and simplifying the resulting inequality.    
\end{proof}

To continue the analysis, we define and analyze the following function $f(\lambda, j, q)$, that we will use repeatedly, and is related to the terms of binomial sums. 

\begin{definition}
\label{def:f}
For $\lambda\leq n$, $j\leq \lambda$ and $q\in (0,1)$, let
$$ f(\lambda, j, q) = {\lambda-1 \choose j-1} q^{j-1} (1-q)^{\lambda-j}$$    
\end{definition}

The function $f(\lambda, j, q)$ equals the probability that a set of exactly $j-1$ users are selected out of  a set of $\lambda-1$ users, according to the randomization procedure of the protocol. We identify some useful properties for the function $f$, which we will also exploit in Section \ref{sec:richer-model}. The following claim is easy to verify.
\begin{claim}
\label{cl:properties}
For the function $f(\lambda, j, q)$, where $\lambda\leq n$, and $j\leq \lambda$, the following hold:
\begin{itemize}
    \item $f(\lambda+1, j, q) = \frac{\lambda(1-q)}{\lambda-j+1} f(\lambda, j, q)$
    \item For every $\lambda> \frac{j-1}{q}$, we have $f(\lambda, j, q)> f(\lambda+1, j, q)$.
    \item For every $\lambda < \frac{j-1}{q}$, we have $f(\lambda, j, q) < f(\lambda+1, j, q)$.
    \item For $\lambda = \frac{j-1}{q}$ (if this is an integer), we have $f(\lambda+1, j, q) = f(\lambda, j, q)$.
\end{itemize}
\end{claim}

In particular, by Definition \ref{def:f}, we can see that the first two equilibrium constraints \eqref{eq:new'from-C-to-G} and \eqref{eq:new'from-G-to-C}, can be rewritten using the $f$ function as:
\begin{eqnarray}
f(\lambda, k, q) & \geq & \frac{\beta}{r} ~~ \label{eq:new}  \\
 f(\lambda+1, k, q) & \leq  & \frac{\beta}{r} ~~ \label{eq:new'} 
\end{eqnarray}

We pay particular attention to these two constraints, as they already allow us to conclude on how many contributors there can be at an equilibrium where both contributors and free riders are present. The next two lemmas highlight that once we are given the parameters $n, k, q$, we cannot have equilibria with many different values for $\lambda$. 
Namely, with the exception of some corner cases, there can only be a single value of $\lambda$ for equilibria that contain both contributors and free riders.

\begin{lemma}
Given a participation game with retraction, there can be at most one value for $\lambda\in [\frac{k-1}{q}, n)$ that satisfies the equilibrium constraints \eqref{eq:new'from-C-to-G} and \eqref{eq:new'from-G-to-C} both with strict inequality. 
\end{lemma}

\begin{proof}
Suppose that there exists a value for $\lambda$, with $n> \lambda \geq \frac{k-1}{q}$, such that both \eqref{eq:new'from-C-to-G} and \eqref{eq:new'from-G-to-C} are satisfied with strict inequality. 
Let us consider whether the value $\lambda+x$ for some integer $x\geq 1$, can satisfy the equilibrium constraints. 
Constraint \eqref{eq:new'from-C-to-G} for $\lambda+x$ can be written as: $f(\lambda+x, k, q) \geq \frac{\beta}{r}$.
But by our assumptions, we have that $f(\lambda+1, k, q) < \frac{\beta}{r}$, and by using Claim \ref{cl:properties}, we have:
$$ f(\lambda +x, k, q) \leq f(\lambda+1, k, q) < \frac{\beta}{r}$$
Hence, we reached a contradiction.

In a similar manner, let us consider whether the value $\lambda-x$ for some integer $x\geq 1$ (and with $\lambda-x \geq \frac{k-1}{q}$), can satisfy the equilibrium constraints. By our assumptions, and by Claim \ref{cl:properties}, we have that $f(\lambda-x, k, q) \geq f(\lambda, k, q) > \frac{\beta}{r}$.
If the value $\lambda-x$ satisfies the constraints, it has to satisfy \eqref{eq:new'from-G-to-C}, which is equivalent to $f(\lambda-x+1, k, q) \leq \frac{\beta}{r}$, a contradiction, since $f(\lambda-x+1, k, q) \geq f(\lambda, k, q) > \frac{\beta}{r}$.
\end{proof}

The next lemma deals with the corner case where we have exact equality in a constraint.
\begin{restatable}{rlm}{maxtwoval}
\label{lem:2values}
Suppose that there exists $\lambda\in [\frac{k-1}{q}, n)$ with $f(\lambda, k, q) = \beta/r$. Then, the constraints \eqref{eq:new'from-C-to-G} and \eqref{eq:new'from-G-to-C} are either satisfied both for $\lambda$ and $\lambda-1$ or for $\lambda$ and $\lambda+1$, but for no other values in $[\frac{k-1}{q}, n)$.
\end{restatable}

The proof is deferred to Appendix \ref{appendixproofasymnosameclass}. Using now  Claim \ref{cl:properties} again, we can identify the range of $\beta/r$ that is necessary for an equilibrium to exist.



\begin{lemma}
\label{lem:beta-range}
    Given a game, there exists a value for $\lambda$ that satisfies the constraints \eqref{eq:new'from-C-to-G} and \eqref{eq:new'from-G-to-C}, if and only if  $\frac{\beta}{r} \in [f(n-1, k, q), f(\lceil \frac{k-1}{q}\rceil, k, q)]$.
\end{lemma}

\begin{proof}
    Recall that we need $\lambda \geq \frac{k-1}{q}$, for an equilibrium to exist, and we also know that $f$ is decreasing when $\lambda$ satisfies this lower bound, by Claim \ref{cl:properties}. We can think of the function $f$ when $\lambda$ varies from $\lceil \frac{k-1}{q}\rceil$ to $n-1$, as creating the subintervals $[f(n-1, k, q), f(n-2, k, q)]$, $[f(n-2, k, q), f(n-3, k, q)]$, $\dots$, $[f(\lceil \frac{k-1}{q}\rceil +1, k, q) ,f(\lceil \frac{k-1}{q}\rceil, k, q)]$.
    Hence, if $\beta/r$ belongs to the range stated in the lemma, it belongs to one of these subintervals. And this means that there exists $\lambda$ such that $\beta/r \in [f(\lambda+1, k, q), f(\lambda, k, q)]$. But this precisely means that the constraints \eqref{eq:new'from-C-to-G} and \eqref{eq:new'from-G-to-C} are satisfied with this $\lambda$. We note that it is also possible that there are two consecutive values for $\lambda$ that can satisfy the constraints if $\beta/r$ is equal to one of the endpoints of the subintervals, as described in Lemma \ref{lem:2values}. 
\end{proof}


The next step is to understand the range of $\alpha/r$ for which there exists an equilibrium. This means that we have to deal now with the constraints \eqref{eq:new'from-C-to-L} and 
\eqref{eq:new'-L-and-G}. Note that both constraints imply an upper bound on $\alpha/r$, and  \eqref{eq:new'from-C-to-L} implies a dependence of this upper bound on $\beta$. 
Hence, by taking the minimum of these two bounds, and by abbreviating the binomial terms using the function $f$, we can conclude with the following result:

\begin{theorem}
\label{thm:C-and-F}
    Given a participation game with retraction, there exists a non-trivial equilibrium with $\lambda$ contributors and $n-\lambda$ free riders if and only if the following conditions are met:
    \begin{itemize}
        \item $n-1\geq \lambda \geq \frac{k-1}{q}$
        \item $\frac{\beta}{r} \in [f(\lambda+1, k, q), f(\lambda, k, q)]$
        \item $0 \leq \frac{\alpha}{r} \leq \min\{ q\cdot \sum_{j=k-1}^{\lambda-1}f(\lambda, j+1, q) -q\frac{\beta}{r}, ~q\cdot \sum_{j=k}^{\lambda} f(\lambda+1, j+1, q)\}$
    \end{itemize}
    Moreover, whenever the above constraints are satisfied, there is either a unique value for $\lambda$ or two consecutive values that can satisfy them.
\end{theorem}


We illustrate Theorem \ref{thm:C-and-F} with the following example. 
\begin{example}
Suppose $k=13, q=0.3$, and $n=60$. 
Note that $\frac{k-1}{q} = 40$. Hence, a necessary condition to have non-trivial equilibria is that $\frac{\beta}{r} \in [f(n-1, k, q), f(\frac{k-1}{q}, k, q)] = [0.0355, 0.1366]$. Suppose that we choose $\beta$ and $r$ so that $\frac{\beta}{r} = 0.1$. Then by looking at the range of $f$, we can verify that $\lambda$ should be equal to $\lambda = 48$, i.e. one can see that $\frac{\beta}{r} \in [f(49, k, q), f(48, k, q)]$. Then, by looking at the constraints for $\alpha$ we conclude that if $\frac{\alpha}{r}\leq 0.1382$, there exists a  non-trivial equilibrium with 48 contributors among the 60 players. Thus, in this game, if we set the reward appropriately (roughly 10 times more than each of the cost parameters), we have an equilibrium with high participation.  

\end{example}


\paragraph{The all-in profile.} To finish with the analysis of this section, we also need to check if we can have an equilibrium where everybody participates, i.e., with $|C|=n$. In this case we have a simpler set of constraints and we can identify again precise bounds on the relation between $r$ and the costs $\alpha, \beta$, that should hold. In particular, we only need to utilize Equations \eqref{eq:new'from-C-to-G} and \eqref{eq:new'from-C-to-L}, since there is no other type of players. This implies the following.   

\begin{restatable}{rlm}{allout}
\label{lem:all-out}
    The all-in profile is an equilibrium if and only if the costs $\alpha, \beta$ and the monetary reward $r$ satisfy the conditions below. Furthermore, there always exists a non-empty range for the reward $r$, dependent on the costs $\alpha, \beta$, so that the conditions are satisfied. 

\noindent\begin{minipage}{0.3\linewidth}
    \begin{eqnarray}
\beta/r & \leq & f(n, k, q)   \label{eq:all-in-C-to-F}
\end{eqnarray}
\end{minipage}%
\begin{minipage}{0.1\linewidth}
\end{minipage}%
\begin{minipage}{0.6\linewidth}
    \begin{eqnarray}
\alpha/r  & \leq & q\cdot \left(\sum_{j=k-1}^n f(n, j+1, q) -\beta/r\right) \label{eq:all-in-C-to-A}
\end{eqnarray}
\end{minipage}\par\vspace{\belowdisplayskip}

\end{restatable}

We provide the proof in Appendix \ref{proofallout} for space reasons.

For completeness, we now summarize our findings for the existence of equilibria.

\begin{corollary}
    Consider a participation game with retraction.
    \begin{itemize}
        \item The all-out profile is always an equilibrium (as long as $k>1$).
        \item There exist equilibria with $\lambda$ contributors and $n-\lambda$ free riders only when $\lambda\in \{\lceil\frac{k-1}{q}\rceil,\dots, n-1\}$, and as long as the reward $r$, compared to the cost parameters, satisfies the conditions described in Theorem \ref{thm:C-and-F}. 
        \item There exists a non-empty range for the reward $r$, so that the all-in profile is an equilibrium, as described in Lemma \ref{lem:all-out}.
    \end{itemize}
\end{corollary}

\paragraph{Discussion and comparisons with Section \ref{sec:simple}.} We conclude by highlighting that the equilibrium analysis of games with retraction is significantly different from the simpler games we discussed in Section \ref{sec:simple}. The major difference is that as we saw both in Sections \ref{sec:simple-v=0} and \ref{sec:asym}, players of the same type, i.e., with the same selection probability have to use the same strategy at equilibrium. In contrast, in the richer games we analyzed in this section, players with the same selection probability can utilize different strategies at equilibrium. On the positive side, for both classes of games, we see that beyond the trivial equilibrium, all other equilibria have a relatively high number of contributors, and hence, participation can be incentivized with the use of appropriate rewards.

\section{Participation games with universal payments}
\label{sec:algorand-model}

In this section, we take one more step in the  simplification 
of what information the system should record in order to give rewards as we forfeit the ability of the system to detect even if someone was eligible. 
This means that the protocol rewards any player who decides to participate, as long as progress is made. 
In other words, regardless of whether a player was eligible or not in a given round, she can claim a reward if she simply decided to participate, and if a sufficient amount of players completed their tasks for the system to make progress. 
Obviously, this is a simplest of mechanisms, and all the protocol needs to do is to check who registered to participate. 

Intuitively, reward schemes of this form do not seem to be appropriate choices from the protocol's perspective, as they could end up paying even more free riders. Indeed, we will shortly demonstrate that although such mechanisms may admit equilibria where the job gets done, they can induce more unfair payments to the players who actually did their duty and also result in a higher total expenditure, compared to the models of the previous sections. For brevity in the presentation, we will exhibit our comparisons only for the case where $v=0$. 
In the following exposition, we consider separately the models with and without retraction.\footnote{ We note that both cases are relevant for consideration as they capture the setting when a player cannot avoid any tasks when it participates (no retraction is possible) and when they can avoid some tasks (possibility of retraction in order to decrease costs).}

\subsection{Games with universal payments without the possibility of retraction}

Consider the following change in the model of Section \ref{sec:simple}: rewards are given to all users who declared participation as long as progress is made, and regardless of their subsequent eligibility. The possible scenarios that may occur can be seen in Table \ref{tab:tab-advance-simple}.
Hence, with $v=0$, and under a strategy profile $s = (s_1,\dots, s_n)$, the utility of a player $i$, who has chosen to participate  is
$$ u_i(s) = r\cdot Pr[\mbox{progress is made}] - \alpha = r\cdot p(s) - \alpha $$

\begin{table}[tbp]
\begin{center}
\begin{tabular}{|c|c|c|}
\hline
Possible scenarios & Progress is made & No Progress \\ \hline
Abstain & $v$ & $0$   \\ \hline
Participate (regardless of eligibility)  & $r + v - \alpha$ & $- \alpha$ \\ \hline
\end{tabular}
\caption{Possible events and corresponding rewards in a participation game with universal payments.} \label{tab:tab-advance-simple}
\end{center}
\end{table}

In analogy to our results in Section \ref{sec:simple-v=0}, we also obtain here that there can be at most 2 equilibria, the all-in and the all-out profile.

\begin{theorem}[Characterization]
\label{thm:simple-advance}
When $v=0$, we can have at most two Nash equilibria as follows:
\begin{itemize}
    \item The trivial (all-out) profile $(\lazy, \lazy, \lazy)$, where nobody contributes, is a pure Nash equilibrium for $k>1$, or when $k=1$ and $q\leq {\alpha \over r}$.
    \item There is no equilibrium that has both a positive number of contributors and a positive number of abstainers.
    \item The all-in profile, where everybody participates, is an equilibrium if and only if:
\begin{equation}
\label{eq:advance-simple}
 \sum_{j=k}^{n}{n \choose j} q^{j} (1-q)^{n-j} \geq \frac{\alpha}{r} \end{equation}
\end{itemize}
\end{theorem}

\begin{proof}
It is trivial to see that the all-out profile is an equilibrium. Furthermore, for the all-in profile to be an equilibrium, it suffices to check what values of the reward $r$ do not provide incentives for a player who participates to abstain. From the definition of the game, this is true precisely when the probability of making progress with $n$ participants is at least $\alpha/r$. But this is equivalent to Equation \eqref{eq:advance-simple}.

It remains to show that we cannot have any other equilibrium.
For the sake of contradiction, suppose that there is another equilibrium profile $s$, with $\lambda$ contributors and $n-\lambda$ users opting out, where $0< \lambda < n$. In order to have such an equilibrium, it should hold that abstainers have no incentive to participate and the contributors also do not have an incentive to abstain. These two constraints correspond to the following two inequalities.

\begin{eqnarray}
\sum_{j=k}^{\lambda}{\lambda \choose j} q^{j} (1-q)^{\lambda-j} & \geq & \frac{\alpha}{r}  \label{eq:advance_from-C-to-A} \\
\sum_{j=k}^{\lambda+1}{\lambda+1 \choose j} q^{j} (1-q)^{\lambda+1-j} & \leq & \frac{\alpha}{r} \label{eq:advance_from-A-to-C} 
\end{eqnarray}

But we can now apply Lemma \ref{lem:no-intermediate} for the LHS of \eqref{eq:advance_from-A-to-C} (using the value of $\lambda+1$ instead of $\lambda$). This directly establishes that we cannot satisfy simultaneously Equations \eqref{eq:advance_from-C-to-A} and \eqref{eq:advance_from-A-to-C} and hence the proof is complete.

\end{proof}

Moving on, we would like to 
compare the game of this section, against the original game in Section \ref{sec:simple-v=0}. 
We will do this comparison in terms of the total expenditure needed for the protocol to have the all-in profile as an equilibrium, since there can be no other non-trivial equilibria. Clearly, in both games, given the characterizations of Theorem \ref{thm:simplemodel} and Theorem \ref{thm:simple-advance}, if we make the reward $r$ large enough, we can enforce that the all-in profile is an equilibrium in both games. But for any fixed such $r$, the total expenditure for the game with universal payments would be $n\cdot r$, whereas the original game will have a total expected expenditure of $q\cdot n\cdot  r$, which is strictly smaller. Furthermore, it is also natural to focus on the minimum reward needed in each game, so as to have the all-in profile as an equilibrium. Again the comparison yields a higher total expenditure for the model with universal payments, which is shown in the following corollary.

\begin{corollary} 
\label{cor:simple-advance}
Suppose that $q<1$.
\begin{itemize}
\item[(i)] Fix a common value for the reward $r$ such that the all-in profile is an equilibrium both in the original game and in the game with universal payments. Then the total expenditure in the game with universal payments (which is $n\cdot r$) is strictly higher than the total expected expenditure in the original game, which is equal to $q\cdot n\cdot  r$.
\item[(ii)] Let $r_{min}$ (resp. $r_{min}'$) be the minimum possible reward that makes the all-in profile an equilibrium in the original game (resp. in the game with universal payments). 
Then, the total expenditure is strictly higher in the game with universal payments, under these reward schemes.
    \end{itemize}
\end{corollary}

\begin{proof}
    The proof of (i) follows from the discussion before the statement of the corollary. For (ii), let $s^* = (P,\dots, P)$ be the all-in profile where everybody chooses to participate. Consider first the original game of Section \ref{sec:simple}. From Theorem \ref{thm:simplemodel}, it follows that
    $$ r_{min} = \frac{\alpha}{q\cdot p(i, s^*_{-i})}$$
    The above expression holds for any $i$ (by symmetry, it does not make a difference which player $i$ we use). For the universal payments, the minimum viable reward to make the all-in profile an equilibrium is implied by Theorem \ref{thm:simple-advance}, and is equal to 
    $$ r_{min}' = \frac{\alpha}{p(s^*)} = \frac{\alpha}{q\cdot p(i, s^*_{-i}) + (1-q)\cdot p(s^*_{-i})}$$
    
    Clearly, $r_{min}' < r_{min}$, when $q<1$. The total expenditure in the universal payments game is $n\cdot r_{min}'$. On the other hand the expected expenditure at the original model is 
    $$q\cdot n\cdot r_{min} = n\cdot \frac{\alpha}{p(i, s^*_{-i})} \leq n\cdot \frac{\alpha}{q\cdot p(i, s^*_{-i}) + (1-q)\cdot p(s^*_{-i})} $$

The last inequality above holds because $p(i, s^*_{-i}) >  p(s^*_{-i})$. This completes the proof.
\end{proof}

We end this subsection by highlighting one more negative aspect of the mechanism with universal payments. As seen within the proof of Corollary \ref{cor:simple-advance}, it holds that $r_{min}' < r_{min}$. This means that we can make the all-in profile an equilibrium using a smaller reward per user, under the universal payment scheme. Therefore, not only the protocol has a higher total expenditure, but it also gives less rewards to the people who actually contributed to make progress in comparison to the original game.

\subsection{Games with universal payments and retraction}
\label{sec:univeral+retraction}

In this subsection, our goal is to produce analogous conclusions to Corollary \ref{cor:simple-advance} for the games with retraction. 
Hence, consider the game defined in Section \ref{sec:richer-model}, and suppose that we modify it by having universal payments for anybody who chose to participate, regardless of eligibility. The possible scenarios that can occur and the corresponding utility are shown in Table~\ref{tab:tab-events-algorand}.  

\begin{table}[tbp]
\begin{center}
\begin{tabular}{|c|c|c|}
\hline
Scenarios & Progress is made & No Progress \\ \hline
Abstain & $v$ & $0$   \\ \hline
Declares participation, not eligible & $r + v - \alpha$ & $-\alpha$ \\ \hline
Declares participation, eligible and completes tasks  & $r+v - \alpha-\beta$ & $- \alpha-\beta$ \\ \hline
Declares participation, eligible and retracts & $r+v - \alpha$ & $- \alpha$ \\ \hline
\end{tabular}
\caption{Possible events and corresponding rewards in a participation game with retraction and universal payments.} \label{tab:tab-events-algorand}
\end{center}
\end{table}
 
Given a profile $s_{-i}$ for all players except $i$, the expected utility of player $i$, for each one of her pure strategies is described in Table~\ref{tab:tab-utilities-algorand}.

\begin{table}[tbp]
\begin{center}
\begin{tabular}{|l|c|}
\hline
Action of player $i$ & Expected utility of pl. $i$, given $s_{-i}$   \\ \hline
Abstain & $p(s_{-i})v$    \\ \hline
Participate, contribute if eligible & $(1-q)p(s_{-i})(r+v) + q[p(i, s_{-i})(r+v) -\beta] - \alpha$  \\ \hline
Participate, don't contribute & $p(s_{-i})(r+v) - \alpha$\\ \hline
\end{tabular}
\caption{Expected utility under the possible events for a player $i$.} \label{tab:tab-utilities-algorand}
\end{center}
\end{table}

As in the analysis of Section \ref{sec:richer-model}, we can discount the event that there exist equilibria with contributors and abstainers (cf. Theorems \ref{thm:no-C-and-A} and  \ref{thm:simple-advance}.
)

We write the equilibrium constraints, that need to hold for an equilibrium with $\lambda$ contributors and $n-\lambda$ free riders, where $\lambda>0$ and $\lambda < n$. The calculations yield the following system.

\begin{eqnarray}
f(\lambda, k, q) & \geq & \frac{\beta}{r} ~~ \label{eq:algorand-from-C-to-G}  \\
f(\lambda+1, k, q) & \leq & \frac{\beta}{r} ~~ \label{eq:algorand-from-G-to-C} \\
r\cdot \left( q\cdot f(\lambda, k, q) +  \sum_{j=k}^{\lambda-1}f(\lambda, j+1, q)\right) & \geq & \alpha + \beta q ~~ \label{eq:algorand-from-C-to-L}  \\
r\cdot \sum_{j=k}^{\lambda}f(\lambda+1, j+1, q)& \geq & \alpha ~~ \label{eq:algorand-from-G-to-L} 
\end{eqnarray}

If we substitute the terms in the above equations, this boils down to the following system.

\begin{eqnarray}
{\lambda-1 \choose k-1} q^{k-1} (1-q)^{\lambda-k}& \geq & \frac{\beta}{r} ~~ \label{eq:new-algorand-from-C-to-G}  \\
{\lambda \choose k-1} q^{k-1} (1-q)^{\lambda-k+1}& \leq & \frac{\beta}{r} ~~ \label{eq:new-algorand-from-G-to-C} \\
r\cdot \left( q\cdot {\lambda-1 \choose k-1} q^{k-1} (1-q)^{\lambda-k} + \sum_{j=k}^{\lambda-1}{\lambda-1 \choose j} q^{j} (1-q)^{\lambda-1-j} \right)& \geq & \alpha + \beta q ~~ \label{eq:new-algorand-from-C-to-L}  \\
r\cdot \sum_{j=k}^{\lambda}{\lambda \choose j} q^{j} (1-q)^{\lambda-j}& \geq & \alpha ~~ \label{eq:new-algorand-from-G-to-L} 
\end{eqnarray}

As in the previous section, we can again establish that Claim \ref{cl:lambda-bound} holds, and hence
$$n-1 \geq \lambda \geq \frac{k-1}{q} $$

It is now easy to verify that we can have at least as many equilibria as in the previous section, with contributors and free riders, as per the following corollary.
\begin{corollary}
\label{cor:algorand}
    If a strategy profile with $\lambda$ contributors and $n-\lambda$ free riders is an equilibrium for the game of Section \ref{sec:richer-model}, then it is also an equilibrium for the universal payments game.
\end{corollary}
\begin{proof}
Consider such an equilibrium. The first two constraints are identical with the previous section and are automatically satisfied. As for the last two constraints, we can see that they are relaxed versions of the corresponding constraints of Section \ref{sec:richer-model}, and hence they will be satisfied as well.
\end{proof}

Finally, if we follow the analysis of the previous section, 
we can arrive at a similar characterization theorem, which we state below.
\begin{theorem}
\label{thm:algorand-C-and-F}
    Given a universal payments game, there exists a non-trivial equilibrium with $\lambda$ contributors and $n-\lambda$ free riders if and only if the following conditions are met:
    \begin{itemize}
        \item $n-1\geq \lambda \geq \frac{k-1}{q}$
        \item $\frac{\beta}{r} \in [f(\lambda+1, k, q), f(\lambda, k, q)]$
        \item $0 \leq \frac{\alpha}{r} \leq \min\{ q\cdot(f(\lambda, k, q) -\frac{\beta}{r}) +  \sum_{j=k}^{\lambda-1} f(\lambda, j+1, q), ~ \sum_{j=k}^{\lambda} f(\lambda+1, j+1, q)\}$
    \end{itemize}
    Moreover, whenever the above constraints are satisfied, there is either a unique value for $\lambda$ or two consecutive values that can satisfy them.
\end{theorem}

\begin{corollary} 
\label{cor:retract-advance}
Fix an integer $\lambda$ with $\lambda\leq n$, and suppose $q<1$. Let $r_{min}$ (resp. $r_{min}'$) be the minimum possible reward that makes the profile with $\lambda$ contributors and $n-\lambda$ free riders an equilibrium in the original game of Section \ref{sec:richer-model} (resp. in the game with universal payments). Then, the total expenditure is higher in the game with universal payments.

\begin{proof}

We will show that the total expenditure $T'=n\cdot r_{min}'$ in the universal payments game is higher than the total expenditure $T=q\cdot n\cdot r_{min}'$ in the original one. We do so by comparing the bounds for $r_{min}$ and $r_{min}'$ derived by Theorems \ref{thm:C-and-F} and \ref{thm:algorand-C-and-F} and showing that in all cases $r_{min}' \geq r_{min}\cdot q$, which produces the claimed relation for $T,T'$. More concretely, Theorem \ref{thm:C-and-F} implies that in an equilibrium with $\lambda$ contributors,  $r_{min}=\max\left\{r_1,r_2,r_3\right\}$ where

\begin{eqnarray}
         r_1 &=& \frac{\beta}{f(\lambda, k, q)} \label{rbound1}\\
         r_2 &=&  \frac{\alpha}{q\cdot \sum_{j=k-1}^{\lambda-1}f(\lambda, j+1, q) -q\frac{\beta}{r_{min}}} \label{rbound2}\\
        r_3 &=&  \frac{\alpha}{q\cdot \sum_{j=k}^{\lambda} f(\lambda+1, j+1, q)} \label{rbound3}
\end{eqnarray}

and Theorem \ref{thm:algorand-C-and-F} implies implies that in an equilibrium with $\lambda$ contributors, $r_{min}'=\max\left\{r'_1,r'_2,r'_3\right\}$ where

\begin{eqnarray}
         r'_1 &=& \frac{\beta}{f(\lambda, k, q)} \label{rpbound1}\\
        r'_2 &=&  \frac{\alpha}{ q\cdot\left(f(\lambda, k, q) -\frac{\beta}{r'_{min}}\right) +  \sum_{j=k}^{\lambda-1} f(\lambda, j+1, q)} \label{rpbound2}\\
        r'_3 &=&  \frac{\alpha}{ \sum_{j=k}^{\lambda} f(\lambda+1, j+1, q)} \label{rpbound3}
\end{eqnarray}

We compare the bounds in pairs. We first point out that $r_1=r'_1$ and thus $r_{1}' \geq r_{1}\cdot q$ as $1\geq q$, 
and that $r_{3}' = r_{3}\cdot q$. 
It remains to show that  $r_{2}' \geq r_{2}\cdot q$.

We have that  
\begin{eqnarray}
r_{2}\cdot q&=&\frac{\alpha}{ \sum_{j=k-1}^{\lambda-1}f(\lambda, j+1, q) -\frac{\beta}{r}} \\
&=& \frac{\alpha}{\left(f(\lambda, k, q) -\frac{\beta}{r_{min}} \right) + \sum_{j=k}^{\lambda-1}f(\lambda, j+1, q) }
\end{eqnarray}

Thus, we have that $r_{2}' \geq r_{2}\cdot q$ is true if and only if $f(\lambda, k, q) -\frac{\beta}{r_{min}}  \geq q\cdot\left(f(\lambda, k, q) -\frac{\beta}{r'_{min}}\right)$. We can rewrite this as





\begin{eqnarray}
f(\lambda, k, q)\cdot(1-q) \geq \frac{\beta}{r_{min}}  - q\cdot\frac{\beta}{r'_{min}} \\
f(\lambda, k, q)\cdot(1-q) \geq \beta \cdot \frac{  r'_{min} - q \cdot r_{min}}{r_{min}\cdot r'_{min}} \label{trickyRHS}
\end{eqnarray}

We now proceed with a proof by cases for the right hand side of eq. \ref{trickyRHS}. If it is negative, then eq. \ref{trickyRHS} is true as $f$ and $1-q$ are both positive. This in turn proves that $r_{2}' \geq r_{2}\cdot q$. We know have that $r_{i}' \geq r_{i}\cdot q$ for $i \in \{1,2,3\}$ which implies that $r_{min}' \geq r_{min}\cdot q$ and thus completes the proof.

If the right hand side of eq. \ref{trickyRHS} is non-negative, it must be that  $  r'_{min} - q \cdot r_{min}\geq 0$ as $\beta$ and $r_{min},r'_{min}$ are all positive. This completes the proof as it shows that  $  r'_{min} \geq q \cdot r_{min}$ directly.
 
\end{proof}

\end{corollary}


\section{Conclusions and future research}

We have demonstrated that by carefully setting reward levels we can achieve equilibria with high participation in blockchain participation games, even in cases where the bookkeeping of the underlying system does not track user behavior in detail. This is beneficial as high participation in turn increases the resiliency of such systems. Whilst we believe our results are encouraging, we also wish to highlight a number of avenues for future research.  

\begin{enumerate}
    \item It would be interesting to enrich the model in terms of different rewards and operational costs between players.
    \item In real world blockchains, participation games are played repeatedly. As such it would be natural to also treat the repeated version of the games we model, either theoretically or via experimental approaches. If some participants can observe the behaviors and contributions of others before committing to their own decisions, it would be interesting to analyze best response dynamics in this setting. 
    \item Again, in the real world, blockchain participants are human beings subject to complex motivations. We believe it would be beneficial to leverage behavioral (experimental) game theory in the setting of participation games. 
    \item Our model treats the underlying blockchain  as black-box without considering the exploitation of protocol specific deviations. In a real world deployment however, strategic participants may choose to participate and deviate in protocol-specific ways (e.g., by engaging in mining games \cite{mininggame}) -- hence it is interesting to study the games that result in this setting.
\end{enumerate}
\ignore{
\noindent \pc{tried to compress into last bullet} {\bf Temporal dimension.} We can assume that the game proceeds in discrete time steps, so that each player can decide to check eligibility in any time step she desires, and most importantly, she can observe at any time step, what the others have done so far. To start with a simpler case, let us assume that we have just 2 time steps, say the {\it early} phase and the {\it late} phase. The players who decide to act during time step 1 could be the players who might care more for participating or for making progress. The players who decide to act in time step 2 can be thought of as players who want to check first what happened during time step 1. NOTE: Players at time step 2 can only see how many signed during step 1, but not how many checked eligibility. }

\bibliographystyle{plain}
\bibliography{free-riding}

\newpage
\appendix

\section{Missing proofs from Section 2}

\subsection{Proof of Lemma \ref{lemma-asymnosameclass}} \label{appendixproofasymnosameclass}
\asymnosameclass*
\begin{proof}

With no loss of generality, we assume $a\in C$ and $b\in A$. Following the notation of Lemmas \ref{lemma-asymcontrib}, \ref{lemma-asymnoncon}, it must be that $q_a \cdot L_a \geq {\alpha\over r}$ and $q_b \cdot M_b \leq {\alpha\over r}$. We will show that in fact $M_b > L_a$, leading to a contradiction.

    
    We start by rewriting $M_b$  as:
\begin{align*}
    &\sum_{m=k-1}^{\lambda} \sum_{S\subseteq C, |S|=m} &~\left( \prod_{j\in S} q_j \cdot \prod_{j\in (C\setminus S)} (1-q_j) \right) &=&\\
    &\sum_{m=k-1}^{\lambda} \sum_{S\subseteq C, |S|=m,a\in S} &~\left( \prod_{j\in S} q_j \cdot \prod_{j\in (C\setminus S)} (1-q_j) \right) &&\\
   + &\sum_{m=k-1}^{\lambda} \sum_{S\subseteq C, |S|=m,a\notin S} &~\left( \prod_{j\in S} q_j \cdot \prod_{j\in (C\setminus S)} (1-q_j) \right) &=&\\
 &\sum_{S\subseteq C, |S|=k-1,a\in S} &~\left( \prod_{j\in S} q_j \cdot \prod_{j\in (C\setminus S)} (1-q_j) \right)&& \quad \mbox{(we isolate the $k-1$ term.)}\\
   +  &\sum_{m=k}^{\lambda} \sum_{S\subseteq C, |S|=m,a\in S} &~\left( \prod_{j\in S} q_j \cdot \prod_{j\in (C\setminus S)} (1-q_j) \right) &&\\
   + &\sum_{m=k-1}^{\lambda-1} \sum_{S\subseteq C, |S|=m,a\notin S} &~\left( \prod_{j\in S} q_j \cdot \prod_{j\in (C\setminus S)} (1-q_j) \right) &= &\quad \mbox{(the term for $\lambda$ is nil.)}&\\ 
&\sum_{S\subseteq C, |S|=k-1,a\in S} &~\left( \prod_{j\in S} q_j \cdot \prod_{j\in (C\setminus S)} (1-q_j) \right)&& \mbox{(we will rewrite this as $\Sigma_3$)} \\  
        + q_a \cdot &\sum_{m=k-1}^{\lambda-1} \sum_{S\subseteq C\setminus \{a\}, |S|=m} &~\left( \prod_{j\in S} q_j \cdot \prod_{j\in (C\setminus \{i\}\setminus S)} (1-q_j) \right) &&\\
   + (1-q_a) \cdot&\sum_{m=k-1}^{\lambda-1} \sum_{S\subseteq C\setminus \{a\}, |S|=m} &~\left( \prod_{j\in S} q_j \cdot \prod_{j\in (C\setminus \{i\}\setminus S)} (1-q_j) \right) &=& \Sigma_3 + L_a \quad \mbox{ (where $\Sigma_3>0$)}\\\\  
\end{align*}
\end{proof}

\section{Missing proofs from Section 3}

\subsection{Proof of Lemma \ref{lem:2values}} \label{appendixprooftwoval}

\maxtwoval*

\begin{proof}
     By the definition of $f$, and $\lambda\in [\frac{k-1}{q}, n)$, it must be that $$f(\lambda-1-x, k, q)>f(\lambda-1, k, q)>f(\lambda, k, q)$$ and also it must be that $$f(\lambda, k, q)\geq f(\lambda+1, k, q)>f(\lambda+1+x, k, q)$$ for all $x\geq 1$. The equality is only true when $\lambda = \frac{k-1}{q}$.

     By the first inequality, there can be no equilibria for $\lambda-2$ or below, as both values of $f$ will be greater than $ \beta/r$.  In the case where $\lambda = \frac{k-1}{q}$, we may also discount $\lambda-1$, as $\lambda-1 \notin [\frac{k-1}{q}, n)$ .

     By the second, there can be no equilibria for $\lambda+2$ and above. In the case where $\lambda \neq \frac{k-1}{q}$, the latter bound is improved to $\lambda+1$.

     Thus, at most we can have equilibria for   $\lambda$ and $\lambda+1$ when $\lambda = \frac{k-1}{q}$ and only for $\lambda-1$ and  $\lambda$ otherwise.
\end{proof}
\subsection{Proof of Lemma \ref{lem:all-out}} \label{proofallout}

\allout*
\begin{proof}
    The first statement of the lemma follows directly by utilizing Equations \eqref{eq:new'from-C-to-G} and \eqref{eq:new'from-C-to-L}, which prescribe that no agent has an incentive to become an abstainer or a free rider.
    As for the second claim of the lemma, let $\alpha, \beta$ be the two cost parameters. Then since $f(n, k, q)>0$, we have that there exists a range for the reward $r$ that satisfies Equation \eqref{eq:all-in-C-to-F}. Furthermore, note that the RHS of \eqref{eq:all-in-C-to-A} is also positive. Indeed, since $\beta/r \geq f(n, k, q)$, the RHS of \eqref{eq:all-in-C-to-A} is at least $q\cdot(\sum_{j=k}^n f(n, j+1, q))$. Hence, we can always calibrate $r$ so that both $\alpha/r$ and $\beta/r$ respect their positive upper bounds. 
\end{proof}

\end{document}